\documentclass[12pt,letterpaper]{article}

\usepackage{amsmath}
\usepackage{amssymb}
\usepackage{latexsym} 
\usepackage{array}
\usepackage{enumerate}
      
\usepackage{amsthm}

\newtheorem{theorem}{Theorem}

\begin{document}

\title{\bf BOSONIC PAIRINGS}
\author{P. L. Robinson}
\date{}
\maketitle 
\date{}

\begin{abstract}
We extend the inner product from bosonic Fock space to a pairing between suitable antifunctionals on the symmetric algebra. Our account is illustrated by the (Gaussian) half-form pairing between positive polarizations in the form needed for geometric quantization. 
\end{abstract}

\medbreak 

\section*{Introduction}

\medbreak 

The standard framework in which to discuss a free bosonic system rests on the completion of the symmetric algebra over the underlying one-particle complex Hilbert space. In conventional approaches, one is often faced with figurative expressions which only assume an honest existence when appropriate conditions are satisfied. A specific instance of this scenario occurs in the context of the classical Shale \cite{Shale} theorem on the unitary implementation of a symplectic automorphism in the bosonic Fock representation: a conventional approach to this theorem hinges on the corresponding displaced Fock vacuum, a figurative Gaussian that lies in Fock space precisely when the symplectic automorphism and the one-particle complex structure have Hilbert-Schmidt commutator. In \cite{R boson} we developed a framework in which such figurative objects have a strictly legitimate existence and thereby offered a new account of the Shale theorem. Our purpose here is to show that this framework also supports rigorously extending the Fock space inner product to a pairing between suitable formerly figurative expressions: specifically, between suitable (not necessarily bounded) antifunctionals on the symmetric algebra; such a pairing is important for the theory of geometric quantization, as discussed in \cite{Kostant} and \cite {RR}. 

\medbreak

For convenience, we recall familiar foundational material on bosonic Fock space, perhaps from an unfamiliar perspective. The precise formulation offered here was first presented in \cite{R boson} though a more traditional reference such as \cite{BSZ} or \cite{BraRob} may also be consulted for some of the details. 

\medbreak

Let $V$ be a complex Hilbert space with $\langle \; \vert \; \rangle$ as its inner product. For the sake of simplicity, we shall suppose that $V$ has finite complex dimension $m$. 

\medbreak

The symmetric algebra $SV = \bigoplus_{d \geqslant 0} S^d V$ is graded by degree and carries a canonical inner product relative to which the homogeneous summands are perpendicular, the Fock vacuum $\mathbf{1} \in S^0 V = \mathbb{C}$ is a unit vector and if $x_1, \dots , x_d, y_1, \dots ,y_d \in V$ then 
	\[\langle x_1 \cdots x_d \vert y_1 \cdots y_d \rangle = \sum_{p} \prod_{j = 1}^{d} \langle x_j \vert y_{p(j)} \rangle 
\]
where $p$ runs over all permutations of $\{1, \dots ,d\}$. In particular, if $x, y \in V$ then 
	\[\langle x^d \vert y^d \rangle = d! \: \langle x \vert y \rangle^d 
\]
and if $V$ has unitary basis $(v_1, \dots , v_m)$ then $SV$ has unitary basis $\{v^D : D \in {\mathbb N}^m \}$ where if $D = (d_1, \dots ,d_m) \in {\mathbb N}^m$ then 
	\[v^D = \frac{v_1 ^{d_1} \cdots v_m ^{d_m}}{\sqrt{d_1! \cdots d_m!}}\: .
\]

\medbreak

Let $SV'$ be the full (purely algebraic) antidual of the symmetric algebra, comprising all (not necessarily bounded) antilinear functionals $SV \rightarrow \mathbb{C}$. The antidual $SV'$ is naturally a commutative associative complex algebra: its product is defined by the rule that if $\Phi, \Psi \in SV'$ then 
	\[\theta \in SV \Longrightarrow [\Phi \Psi](\theta) = [\Phi \otimes \Psi](\Delta \theta)
\]
where the cocommutative coproduct $\Delta : SV \rightarrow SV \otimes SV$ is the composite $SV \rightarrow S(V \oplus V) \rightarrow SV \otimes SV$ in which the first map is induced by the diagonal $V \rightarrow V \oplus V$ and the second is the canonical isomorphism. Note that the grading on $SV$ gives each $\Phi \in SV'$ the structure of a formal series, thus
	\[\Phi = \sum_{d \geqslant 0} \Phi_d
\]
where if $d \geqslant 0$ then $\Phi_d = \Phi \vert S^d V$. In the opposite direction, if to each $d \geqslant 0$ is associated an element $\Phi_d \in S^d V'$ then the formal series above defines an element of $SV'$ because individual elements of $SV$ vanish in sufficiently high degree. 

\medbreak 

The canonical inner product on the symmetric algebra embeds it in its antidual: it is readily verified that the canonical map 
	\[SV \rightarrow SV' : \phi \mapsto \langle \cdot \vert \phi \rangle	
\]
is an injective algebra homomorphism. In these terms, bosonic Fock space $S[V] = \bigoplus_{d \geqslant 0} S^d [V]$ may be defined either as the Hilbert space completion of $SV$ or as the subspace of $SV'$ comprising all bounded antifunctionals. Note that if $\Phi, \Psi \in S[V]$ then their inner product is given by 
	\[\langle \Phi \vert \Psi \rangle = \sum_{d \geqslant 0} \langle \Phi_d \vert \Psi_d \rangle.
\]

\medbreak

A proper treatment of Fock space includes a discussion of the Fock representation in terms of creators and annihilators; as such a treatment is not neccessary for our purposes, we again refer to \cite{R boson}, \cite{BraRob} or \cite{BSZ} for details. 

\medbreak 

\section*{Bosonic pairings}

\medbreak 

Recall that the homogeneous summands in the symmetric algebra $SV$ are mutually perpendicular relative to its standard inner product and that the corresponding remark holds true for the bosonic Fock space $S[V]$. This being so, we are led to define a bosonic pairing $\langle \ : \ \rangle_1$ between suitable elements of the full antidual $SV'$ as follows, at least in preliminary form. For $\Phi, \Psi \in SV'$ we define 
	\[\langle \Phi : \Psi \rangle_1 = \sum_{d \geqslant 0} \langle \Phi_d \vert \Psi_d \rangle
\]
whenever the indicated series is convergent. A little later, we shall extend this definition; for now, we consider properties of the bosonic pairing as defined in the present sense. 

\medbreak

According to our recollection, this pairing extends the standard inner product from $SV$ through $S[V]$ to a partially-defined inner product on $SV'$. The bosonic pairing also reproduces the canonical pairing between $SV$ and $SV'$. 

\medbreak 

\begin{theorem} 
If $\phi \in SV$ and $\Psi \in SV'$ then 
	\[\langle \phi : \Psi \rangle_1 = \Psi (\phi).
\]
\end{theorem}
\begin{proof} 
If $\phi = \sum_{d = 0}^D \phi_d$ then 
	\[\Psi (\phi) = \Psi (\sum_{d = 0}^D \phi_d) = \sum_{d = 0}^D \Psi (\phi_d) = \sum_{d = 0}^D \langle \phi_d \vert \Psi_d \rangle = \sum_{d \geqslant 0} \langle \phi_d \vert \Psi_d \rangle
\]
whence $\langle \phi : \Psi \rangle_1$ exists and has the indicated value. 
\end{proof}

\medbreak

Of course, it is likewise true that if $\Phi \in SV'$ and $\psi \in SV$ then 
	\[\langle \Phi : \psi \rangle_1 = \overline{\Phi (\psi)}.
\]

\medbreak

The bosonic pairing is defined in situations that involve naturally the number operator and its powers. Recall that the number operator $\mathcal{N}$ is defined initially on $SV$ (where it multiplies homogeneous elements by degree) and extends antidually to $SV'$. In Fock space $S[V]$, it is then defined as a self-adjoint operator on the natural domain 
	\[\mathbb{D} (\mathcal{N}) = \{ \Phi \in SV' : \sum_{d \geqslant 0} d^2 \Vert \Phi_d \Vert^2 < \infty \}.
\]
More generally, if $r \in \mathbb{R}$ then the power $\mathcal{N}^r$ has natural domain 
	\[\mathbb{D} (\mathcal{N}^r) = \{ \Phi \in SV' : \sum_{d \geqslant 0} d^{2r} \Vert \Phi_d \Vert^2 < \infty \}
\]
on which it has the effect 
	\[\Phi \in \mathbb{D} (\mathcal{N}^r) \Rightarrow \mathcal{N}^r (\Phi) = \sum_{d \geqslant 0} d^r \Phi_d
\]
with an appropriate understanding of the $d = 0$ term.

\medbreak 

\begin{theorem} 
Let $r \in \mathbb{R}$. If $\Phi \in \mathbb{D} (\mathcal{N}^{-r})$ and $\Psi \in \mathbb{D} (\mathcal{N}^r)$ then 
	\[\langle \Phi : \Psi \rangle_1 = \langle \mathcal{N}^{-r} \Phi \vert \mathcal{N}^r \Psi \rangle .
\]
\end{theorem}
\begin{proof} 
An elementary and direct calculation: if $d > 0$ then 
	\[\langle \Phi_d \vert \Psi_d \rangle = \langle d^{-r}\Phi_d \vert d^r\Psi_d \rangle
\]
whence summation concludes the argument.  
\end{proof}

\medbreak 

The bosonic pairing is also defined in situations of H\"older type. Let $p \geqslant 1$ be real and define 
	\[\mathbb{H}^p [V] = \{ \Phi \in SV' : \sum_{d \geqslant 0} \Vert \Phi_d \Vert^p < \infty \}.
\]
In addition, define $\mathbb{H}^\infty [V]$ to comprise all those $\Phi \in SV'$ for which the sequence $(\Vert \Phi_d \Vert)_{d = 0}^\infty$ is bounded. As a special case, note that $\mathbb{H}^2 [V]$ is precisely Fock space $S[V]$. 

\medbreak 

\begin{theorem} 
Let the indices $p \geqslant 1$ and $q \geqslant 1$ be conjugate in the sense that $\frac{1}{p} + \frac{1}{q} = 1$. If $\Phi \in \mathbb{H}^p [V]$ and $\Psi \in \mathbb{H}^q [V]$ then $\langle \Phi : \Psi \rangle_1$ is defined.
\end{theorem}
\begin{proof} 
Assume that $p > 1$ and $q > 1$. If $d \geqslant 0$ then 
	\[\vert \langle \Phi_d \vert \Psi_d \rangle \vert \leqslant \Vert \Phi_d \Vert \: \Vert \Psi_d \Vert
\]
by the Cauchy-Schwarz inequality, whence summation yields 
	\[\sum_{d \geqslant 0} \vert \langle \Phi_d \vert \Psi_d \rangle \vert \leqslant \{\sum_{d \geqslant 0} \Vert \Phi_d \Vert^p\}^{1/p} \: \sum_{d \geqslant 0} \Vert \Psi_d \Vert^q\}^{1/q}
\]
on account of the H\"older inequality. The case in which $\{ p, q \} = \{ 1, \infty \}$ is still more transparent. 
\end{proof}
	
\medbreak 

As a more specific illustration, let us consider the bosonic pairing of Gaussians. Thus, let the symmetric antilinear map $Z: V \rightarrow V$ and quadratic $\zeta \in S^2 V$ correspond as usual according to the rule 
	\[x, y \in V \Rightarrow \langle y \vert Z x \rangle = \langle x y \vert \zeta \rangle
\]
and consider the associated Gaussian 
	\[e^Z = \sum_{d \geqslant 0} \frac{\zeta^d}{d!} \in SV'.
\]

\medbreak 

The following is Theorem 2.13 in \cite{R boson} but we include a proof of it here for ease of reference. 

\medbreak 

\begin{theorem} \label{norm}
Let $Z : V \rightarrow V$ be symmetric antilinear. The Gaussian $e^Z$ lies in Fock space $S[V]$ precisely when $\Vert Z \Vert < 1$ and then 
	\[\Vert e^Z \Vert^2 = {\rm Det}^{1/2} (I - Z^2)^{-1}.
\]
\end{theorem}
\begin{proof}
By diagonalization, $Z$ furnishes a unitary basis $v_1, \dots, v_m$ for $V$ and nonnegative scalars $\lambda_1, \dots , \lambda_m$ such that if $1 \leqslant k \leqslant m$ then $Z v_k = \lambda_k v_k$ and therefore $\zeta = \frac{1}{2}\sum_{k=1}^m \lambda_k v_{k}^2$. Now, if $d \in \mathbb{N}$ then 
	\[\zeta^d = \sum_{D} \binom{d}{d_1 \cdots d_m} \Bigl(\frac{\lambda_1}{2}\Bigr)^{d_1} \cdots \Bigl(\frac{\lambda_m}{2}\Bigr)^{d_m} \: v_{1}^{2 d_1} \cdots v_{m}^{2 d_m} 
\]
whence 
\[\frac{\Vert \zeta^d \Vert^2}{(d!)^2} = \sum_{D} \binom{2 d_1}{d_1} \cdots \binom{2 d_m}{d_m} \Bigl(\frac{\lambda_1}{2}\Bigr)^{2 d_1} \cdots \Bigl(\frac{\lambda_m}{2}\Bigr)^{2 d_m}
\]
where summation takes place over all multiindices $D = (d_1 , \dots , d_m) \in {\mathbb N}^m$ for which $d_1 + \cdots + d_m = d$. It follows that 
\begin{eqnarray*}
	 \sum_{d \geqslant 0} \frac{\Vert \zeta^n \Vert^2}{(d!)^2} & = & \sum_{d_1 \geqslant 0} \binom{2 d_1}{d_1} \Bigl(\frac{\lambda_1}{2}\Bigr)^{2 d_1} \cdots \sum_{d_m \geqslant 0} \binom{2 d_m}{d_m} \Bigl(\frac{\lambda_m}{2}\Bigr)^{2 d_m} \\ & = & (1 - \lambda_{1}^2)^{-1/2} \cdots (1 - \lambda_{m}^2)^{-1/2} \\ & = & {\rm Det}^{1/2} (I - Z^2)^{-1} 
\end{eqnarray*}
provided that each of the nonnegative numbers $\lambda_1 , \dots ,\lambda_m$ is strictly less than unity. In the opposite direction, if (say) $\lambda = \lambda_k \geqslant 1$ then  
	\[\sum_{d \geqslant 0} \binom{2 d}{d} \Bigl(\frac{\lambda}{2}\Bigr)^{2 d} = \infty
\]
and this divergence already prohibits $e^Z$ from membership in $S[V]$. 
\end{proof}

\medbreak

For convenience, let us denote by $\overline{\mathcal{D}} (V)$ the set of all symmetric antilinear maps $Z : V \rightarrow V$ such that $\Vert Z \Vert \leqslant 1$ and let $\mathcal{D} (V)$ comprise those $Z$ that satisfy $\Vert Z \Vert < 1$; these are versions of closed and open Siegel domains. Let us also write $\mathcal{G} (V)$ for the set comprising all those (necessarily invertible) complex-linear maps $T : V \rightarrow V$ such that if $0 \neq v \in V$ then ${\rm Re} \langle v \vert T v \rangle > 0$: it is readily verified that $\mathcal{G} (V)$ is a convex open neighbourhood of the identity in the complex general linear group on $V$; consequently, $\rm{Det}$ has a holomorphic square-root ${\rm Det}^{1/2}: \mathcal{G} (V) \rightarrow \mathbb{C}$. Now, if $X$ and $Y$ lie in the closed Siegel domain $\overline{\mathcal{D}} (V)$ and $I - Y X$ is invertible then in fact $I - Y X \in \mathcal{G} (V)$: indeed, if $v \in V$ then 
	\[2 {\rm Re} \langle v \vert (I - Y X) v \rangle = (\Vert v \Vert^2 - \Vert Xv \Vert^2) + (\Vert v \Vert^2 - \Vert Yv \Vert^2) + \Vert X v - Y v \Vert^2;
\]
this vanishes only when $X v = Y v$ and $(I - Y^2) v = 0$ which forces $(I - Y X) v = 0$ and therefore $v = 0$ by the supposed invertibility of $I - Y X$. Of course, the invertibility of $I - Y X$ is certainly guaranteed if either $X$ or $Y$ actually lies in the open Siegel domain $\mathcal{D} (V)$. After these preparatory comments, we may evaluate the inner product between a pair of Gaussians in Fock space. 

\medbreak 

\begin{theorem} \label{open}
If $X, Y \in \mathcal{D} (V)$ then 
	\[\langle e^X \vert e^Y \rangle = {\rm Det}^{1/2} (I - Y X)^{-1}.
\]
\end{theorem}
\begin{proof} 
Both sides of the claimed formula are antiholomorphic in $X$ and holomorphic in $Y$. Equality on the diagonal of $\mathcal{D} (V) \times \mathcal{D} (V)$ is established in Theorem \ref{norm}; equality on the whole of $\mathcal{D} (V) \times \mathcal{D} (V)$ follows by the principle of analytic continuation.
\end{proof}

\medbreak

For a variety of purposes (notably for half-form pairings within geometric quantization: see \cite{Kostant} and \cite{RR} for details) this inner product formula must be extended beyond the open Siegel domain. Let $X, Y \in \overline{\mathcal{D}} (V)$ be such that $I - Y X$ is invertible; in this case, recall that $I - Y X \in \mathcal{G} (V)$ and therefore that ${\rm Det}^{1/2} (I - Y X)^{-1}$ is defined. Thus, the right side of the inner product formula in Theorem \ref{open} makes perfectly good sense; by contrast, the left side makes sense as an inner product only when $X, Y \in \mathcal{D} (V)$ for only then do the associated Gaussians lie in Fock space. It is tempting to replace the left side by the bosonic pairing $\langle e^X : e^Y \rangle_1$ and indeed this works when $V$ is one-dimensional. 

\medbreak 

\begin{theorem}
Let $V$ be one-dimensional. If $X, Y \in \overline{\mathcal{D}} (V)$ are such that $I - Y X$ is invertible then 
	\[\langle e^X : e^Y \rangle_1 = {\rm Det}^{1/2} (I - Y X)^{-1}.
\]
\end{theorem} 
\begin{proof} 
For convenience, choose and fix a conjugation $\sigma$ on $V$: thus, $\sigma^2 = I$ and if $x, y \in V$ then $\langle \sigma x \vert \sigma y \rangle = \langle y \vert x \rangle$; further, let $u \in V$ be one of the two unit vectors fixed by the conjugation. There exist complex scalars $a$ and $b$ in the closed unit disc such that $X = a \sigma$ and $Y = b \sigma$ whence $I - Y X = (1 - \overline{a}b)I$ is invertible iff $ \overline{a}b \neq 1$; the quadratics in $S^2 V$ that correspond to $X$ and $Y$ are $\xi = \frac{1}{2} a u^2$ and $\eta = \frac{1}{2} b u^2$ respectively. Now, if $d \in \mathbb{N}$ then $\langle \xi^d \vert \eta^d \rangle = (2 d)! \: (\overline{a}b /4)^d$ whence by summation 
	\[\langle e^X : e^Y \rangle_1 = \sum_{d \geqslant 0} \binom{2 d}{d} (\overline{a}b / 4 )^d = (1 - \overline{a}b )^{-1/2} = {\rm Det}^{1/2} (I - Y X)^{-1}.
\]
Evaluation of the sum in this argument is provided by the general binomial theorem (see \cite{Knopp} Item 247).
\end{proof}

\medbreak

Unfortunately, this na\"ive approach fails when $V$ has complex dimension $m > 1$. To see this by example, let $Z: V \rightarrow V$ be a conjugation and let $\zeta \in S^2 V$ be the corresponding quadratic: in this case, it may be checked that if $d \geqslant 0$ then 
	\[\frac{\Vert \zeta^{d + 1} \Vert^2}{(d + 1)!^2} \Big/\frac{\Vert \zeta^d \Vert^2}{d!^2} = \frac{d + \frac{1}{2}m}{d + 1}
\]
so the formal series 
	\[\langle e^Z : e^{-Z} \rangle_1 = \sum_{d \geqslant 0} (-1)^d \Vert \zeta^d \Vert^2 / d!^2 \geqslant 1
\]
cannot converge, although of course $I + Z^2$ is invertible and ${\rm Det}^{1/2} (I + Z^2)^{-1}$ is defined since $Z^2$ is a positive operator. 

\medbreak

Accordingly, we extend our definition of the bosonic pairing by means of a regularization. As preparation, fix $0 < t < 1$ and define $\langle \; : \; \rangle_t$ by scaling $S^d V$ by $t^d$ for each $d \geqslant 0$: explicitly, define 
	\[\Phi, \Psi \in SV' \Rightarrow \langle \Phi : \Psi \rangle_t = \sum_{d \geqslant 0} \langle \Phi_d \vert \Psi_d \rangle t^{2 d}
\]
whenever the indicated series is convergent. 

\medbreak

\textbf{Definition}: The (extended) bosonic pairing is defined by the rule 
	\[\Phi, \Psi \in SV' \Rightarrow  \langle \Phi : \Psi \rangle = \lim_{t \uparrow 1} \langle \Phi : \Psi \rangle_t
\]
whenever the indicated limit exists. 

\medbreak

We hasten to point out at once that we have indeed fashioned an extension of the original bosonic pairing. 

\medbreak 

\begin{theorem}
Let $\Phi, \Psi \in SV'$. If $\langle \Phi : \Psi \rangle_1$ is defined then $\langle \Phi : \Psi \rangle$ is defined and has the same value.

\end{theorem}
\begin{proof} 
An immediate consequence of the Abel limit theorem for complex power series (see \cite{Knopp} Item 232).
\end{proof}

\medbreak 

Thus, all our previous evaluations of bosonic pairings carry over into this extended context. 

\medbreak

  Standard Tauberian theorems guarantee that where $\langle \: : \: \rangle$ is defined, $\langle \: : \: \rangle_1$ is also defined under certain conditions. The very simplest of these theorems (due to Pringsheim: see Example B27 on page 251 of \cite{Bro} and the Theorem following Item 101 in \cite{Knopp}) implies that if $\Phi \in SV'$ and $\langle \Phi : \Phi \rangle$ is defined then $\langle \Phi : \Phi \rangle_1$ is defined. However, this extended bosonic pairing is a strict extension of the original, as we proceed to see explicitly. 

\medbreak 

In fact, we are now able to pair Gaussians as we would wish. 

\medbreak 

\begin{theorem} 
If $X, Y \in \overline{\mathcal{D}} (V)$ are such that $I - Y X$ is invertible then 
	\[\langle e^X : e^Y \rangle = {\rm Det}^{1/2} (I - Y X)^{-1}.
\]
\end{theorem} 
\begin{proof} 
Let $X, Y$ correspond to the quadratics $\xi, \eta$ in the usual manner. Fix $0 < t < 1$ and notice that $t X, t Y \in \mathcal{D} (V)$: if $d \geqslant 0$ then of course $\langle \xi^d \vert \eta^d \rangle t^{2 d} = \langle (t \xi)^d \vert (t \eta)^d \rangle$ whence summation yields 
	\[\langle e^X : e^Y \rangle_t = {\rm Det}^{1/2} (I - t^2 Y X)^{-1}
\]
according to Theorem \ref{open}. The continuity of ${\rm Det}^{1/2}$ on $\mathcal{G} (V)$ as noted prior to Theorem \ref{open} permits us to let $t \uparrow 1$ and complete the proof. 
\end{proof}

\medbreak

We close with some  remarks on the process of regularization by which we extended the bosonic pairing. On the one hand, it is already well-established as a technique for taming divergent series, under the name A-summability (after Abel: see \cite{Knopp} Sections 59 and 61). On the other hand, it naturally incorporates the grading of the symmetric algebra; indeed, the (extended) bosonic pairing is actually invariant under the corresponding unitary group. Explicitly, let $U_{\mathbb{N}} (SV)$ denote the group comprising all those unitary automorphisms $U$ of $SV$ that respect the grading in that if $d \geqslant 0$ then $U (S^d V) = S^d V$. Naturally, each $U \in U_{\mathbb{N}} (SV)$ acts antidually on $SV'$: thus, if $\Phi \in SV'$ and $\psi \in SV$ then $[U \Phi] (\psi) = \Phi [U^* \psi]$; further, if also $d \geqslant 0$ then $(U \Phi)_d = U (\Phi_d)$.  

\medbreak 

\begin{theorem} 
Let $U \in U_{\mathbb{N}} (SV)$ and let $\Phi, \Psi \in SV'$. If $\langle \Phi : \Psi \rangle$ is defined then so is $\langle U \Phi : U \Psi \rangle$ and 
	\[\langle U \Phi : U \Psi \rangle = \langle \Phi : \Psi \rangle.
\]
\end{theorem}
\begin{proof} 
In fact, if $0 < t < 1$ then plainly $ \langle U \Phi : U \Psi \rangle_t = \langle \Phi : \Psi \rangle_t$ and taking the limit as $t \uparrow 1$ concludes the argument. 
\end{proof}

\medbreak 

In particular, the functorial extension of each element of the unitary group $U(V)$ to a unitary automorphism of $SV \subset SV'$ preserves the bosonic pairing.   

\medbreak 

In contrast, the bosonic pairing is not invariant under the `full' unitary group $U(SV)$ even when $V$ is one-dimensional. To see this, choose a unit vector $v \in V$ and for $d \geqslant 0$ let $v_d = v^d / \sqrt{d!}$ so that $SV$ has $(v_d : d \geqslant 0)$ as unitary basis: in this way, we identify $SV$ with the space $\mathfrak s$ comprising all finitely-nonzero complex sequences and $SV'$ with the space $\mathfrak s'$ comprising all complex sequences; furthermore, if the complex sequences $\lambda = (\lambda_d : d \geqslant 0)$ and $\mu = (\mu_d : d \geqslant 0)$ lie in $\mathfrak s'$ and are bounded then $\langle \lambda : \mu \rangle_t = \sum_{d \geqslant 0} (\overline{\lambda_d}\: \mu_d) t^{2 d}$ when $0 < t < 1$. Now, for example, let $U \in U(SV)$ be defined by fixing $v_0$ and interchanging $v_{2 n -1}$ with $v_{2 n}$ when $n \geqslant 1$; further, let $\lambda_d = 1$ and $\mu_d = (-1)^d$ for $d \geqslant 0$. If $0 < t < 1$ then 
	\[\langle \lambda : \mu \rangle_t = 1 - t^2 + t^4 - t^6 + t^8 - \dots = 1/(1+t^2)
\]
while 
	\[\langle U\lambda : U\mu \rangle_t = 1 + t^2 - t^4 + t^6 - t^8 + \dots = 1 + t^2/(1 + t^2) 
\]
so that $\langle \lambda : \mu \rangle = 1/2$ while $\langle U\lambda : U\mu \rangle = 3/2$. 

\medbreak 

A full investigation of the bosonic pairing, for a one-particle space of arbitrary dimension, promises to be both useful and interesting.

\noindent
Department of Mathematics\\University of Florida\\Gainesville FL 32611
\medbreak
\noindent 
email:   paulr@ufl.edu

\end{document}